\documentclass[sigconf]{acmart}
\usepackage{booktabs} 

\setcopyright{rightsretained}

\usepackage[ruled, noend, noline, linesnumbered]{algorithm2e}

\usepackage{bbm}
\usepackage{multicol}
\usepackage{paralist}
\usepackage{epstopdf}
\usepackage{color}
\usepackage{xcolor}
\usepackage{framed}
\usepackage{csvsimple}
\usepackage{longtable}
\usepackage[none]{hyphenat}
\usepackage{adjustbox}

\usepackage{subfigure}

\newcommand{\mainalg}{{\sc SADDLES}}
\newcommand{\degest}{{\sc DEG}}

\newlength{\dhatheight}

\newtheorem{theorem}{Theorem}[section]
\newtheorem{claim}[theorem]{Claim}

\newcommand{\hN}{\widetilde{N}}
\newcommand{\hNd}{\widetilde{N}(d)}

\newcommand{\ignore}[1]{}

\DeclareRobustCommand*\cal{\@fontswitch\relax\mathcal}

\newcommand{\cD}{\mathcal{D}}

\newcommand{\bone}{\mathbbm{1}}

\newcommand{\eps}{\varepsilon}

\newcommand{\wgt}{\mathrm{wt}}

\newcommand{\EX}{\hbox{\bf E}}

\newcommand{\otilde}{\widetilde{O}}

\newcommand{\aN}{\widehat{N_\Lambda}}

\newcommand{\hindex}{h}
\newcommand{\findex}{z}
\newcommand{\err}{\delta} 
\newcommand{\thresh}{\tau}
\newcommand{\rep}{\ell}

\newcommand{\Sec}[1]{\hyperref[sec:#1]{\S\ref*{sec:#1}}} 
\newcommand{\Eqn}[1]{\hyperref[eqn:#1]{(\ref*{eqn:#1})}} 
\newcommand{\Fig}[1]{\hyperref[fig:#1]{Fig.\,\ref*{fig:#1}}} 
\newcommand{\Tab}[1]{\hyperref[tab:#1]{Tab.\,\ref*{tab:#1}}} 
\newcommand{\Thm}[1]{\hyperref[thm:#1]{Theorem\,\ref*{thm:#1}}} 
\newcommand{\Fact}[1]{\hyperref[fact:#1]{Fact\,\ref*{fact:#1}}} 
\newcommand{\Lem}[1]{\hyperref[lem:#1]{Lemma\,\ref*{lem:#1}}} 
\newcommand{\Prop}[1]{\hyperref[prop:#1]{Prop.~\ref*{prop:#1}}} 
\newcommand{\Prob}[1]{\hyperref[prob:#1]{Prolem~\ref*{prob:#1}}} 
\newcommand{\Cor}[1]{\hyperref[cor:#1]{Corollary~\ref*{cor:#1}}} 
\newcommand{\Conj}[1]{\hyperref[conj:#1]{Conjecture~\ref*{conj:#1}}} 
\newcommand{\Def}[1]{\hyperref[def:#1]{Definition~\ref*{def:#1}}} 
\newcommand{\Alg}[1]{\hyperref[alg:#1]{Alg.~\ref*{alg:#1}}} 
\newcommand{\Ex}[1]{\hyperref[ex:#1]{Ex.~\ref*{ex:#1}}} 
\newcommand{\Clm}[1]{\hyperref[clm:#1]{Claim~\ref*{clm:#1}}} 
\newcommand{\Obs}[1]{\hyperref[obs:#1]{Obs.~\ref*{obs:#1}}} 
\newcommand{\Step}[1]{\hyperref[step:#1]{Step~\ref*{step:#1}}} 

\begin{document} 


\makeatletter
\def\@copyrightspace{\relax}
\makeatother 

\pagestyle{plain}

\title{Provable and Practical Approximations for the Degree Distribution using
Sublinear Graph Samples} \titlenote{Both Talya Eden and Shweta Jain contributed equally 
to this work, and are joint first authors of this work.}

\author {Talya Eden} 
\affiliation{School of Computer Science, Tel Aviv University}  
\affiliation{Tel Aviv, Israel} 
\email{talyaa01@gmail.com}
\author {Shweta Jain} 
\affiliation{University of California, Santa Cruz}
\affiliation{Santa Cruz, CA, USA}
\email{sjain12@ucsc.edu}

\author {Ali Pinar} 
\affiliation{Sandia National Laboratories} 
\affiliation{Livermore, CA} 
\email{apinar@sandia.gov}

\author {Dana Ron} 
\affiliation{School of Computer Science, Tel Aviv University} 
\affiliation{Tel Aviv, Israel} 
\email{danaron@tau.ac.il}
\author {C. Seshadhri} 
\affiliation{University of California, Santa Cruz}
\affiliation{Santa Cruz, CA}
\email{sesh@ucsc.edu}

\begin{abstract} The degree distribution is one of the most fundamental properties used in the
analysis of massive graphs. There is a large literature on \emph{graph sampling},
where the goal is to estimate properties (especially the degree distribution) of a large graph through a small, random sample.
Estimating the degree distribution of real-world graphs poses a significant challenge, due to their heavy-tailed nature
and the large variance in degrees.

We design a new algorithm, \mainalg{}, for this problem, using recent mathematical techniques from the field of \emph{sublinear algorithms}. The \mainalg{} algorithm gives provably accurate outputs for all values of the degree distribution.  For the analysis, we define two fatness measures of the degree distribution, called the \emph{$h$-index} and the \emph{$z$-index}. We prove that \mainalg{} is sublinear in the graph size when these indices are large. A corollary of this result is a provably sublinear algorithm for any degree distribution bounded below by a power law.

We deploy our new algorithm on a variety of real datasets and demonstrate its excellent empirical behavior. In all instances, we get extremely accurate approximations for all values in the degree distribution by observing at most $1\%$ of the vertices.  This is a major improvement over the state-of-the-art sampling algorithms, which typically sample more than $10\%$ of the vertices to give comparable results.  We also observe that the $h$ and $z$-indices of real graphs are large, validating our theoretical analysis.

\end{abstract}

\maketitle

\section{Introduction} \label{sec:intro}
In domains as diverse as social sciences, biology, physics, cybersecurity, graphs are used to represent entities and the relationships
between them. This has led to the explosive growth of network science as a discipline over the past decade. One of
the hallmarks of network science is the occurrence of specific graph properties that are common to varying domains,
such as heavy tailed degree distributions, large clustering coefficients, and small-world behavior. Arguably, the most
significant among these properties is the degree distribution, whose study led to the foundation of network science~\cite{BarabasiAlbert99,BrKu+00,FFF99}.

Given an undirected graph $G$, the degree distribution (or technically, histogram) is the sequence
of numbers $n(1), n(2), \ldots$, where $n(d)$ is the number of vertices of degree $d$. In almost all
real-world scenarios, the average degree is small, but the variance (and higher moments)
is large. Even for relatively large $d$, $n(d)$ is still non-zero, and $n(d)$ typically has 
a smooth non-increasing behavior. In \Fig{intro}, we see the typical degree distribution behavior.
The average degree in a Google web network is less than 10, but the maximum degree is more than 5000.
There are also numerous vertices with all intermediate degrees.
This is referred to as
a ``heavy tailed" distribution. The degree distribution, especially the tail,
is of significant relevance to modeling networks, determining their resilience, spread
of information, and for algorithmics~\cite{CoEr+00,NeStWa01,PeFlLa+02,Ne03,Mi03,ChFa06,SeKoPi11,DuKoPi13,AkKoPi17}. 

With full access to $G$, the degree distribution can be computed in linear time, by simply
determining the degree of each vertex. Yet in many scenarios, we only have \emph{partial}
access to the graph, provided through some graph samples. A naive extrapolation of the degree
distribution can result in biased results. 
The seminal research paper of Faloutsos et al. claimed a power law in the degree
distribution on the Internet~\cite{FFF99}. This degree distribution was deduced by measuring a power law distribution 
in the graph sample generated by a collection of traceroute queries on a set of routers.
%
Unfortunately, it was mathematically and empirically proven that
traceroute responses can have a power law \emph{even if the true network does not}~\cite{LaBy+03,PeRi04,ClMo05,AcClKe+09}. 
In general, a direct extrapolation of the degree distribution from a graph subsample
is not valid for the underlying graph.
This leads to the primary question behind our work.

\emph{How can we provably
and practically 
estimate the degree distribution without seeing the entire
graph?}

\begin{figure*}[h!]
\centering
\subfigure[{\tt amazon0601} copurchase network]{\includegraphics[width=0.23\textwidth]{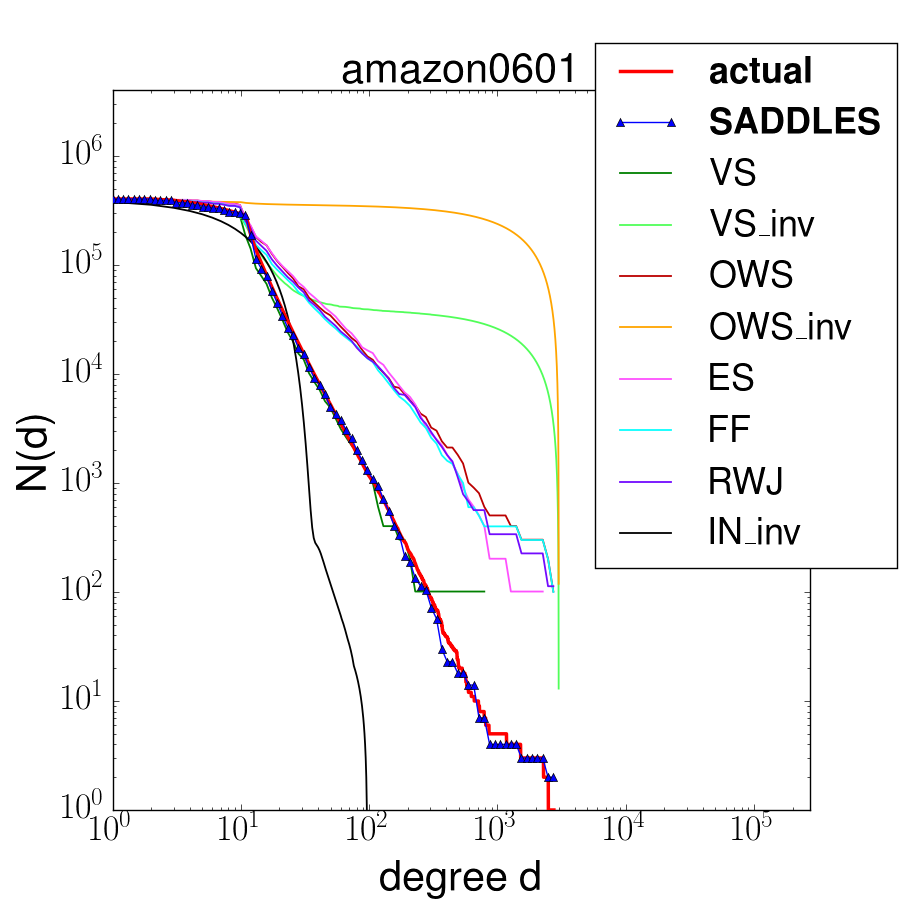}
    \label{fig:intro-amazon}}
\subfigure[{\tt web-Google} web network]{\includegraphics[width=0.23\textwidth]{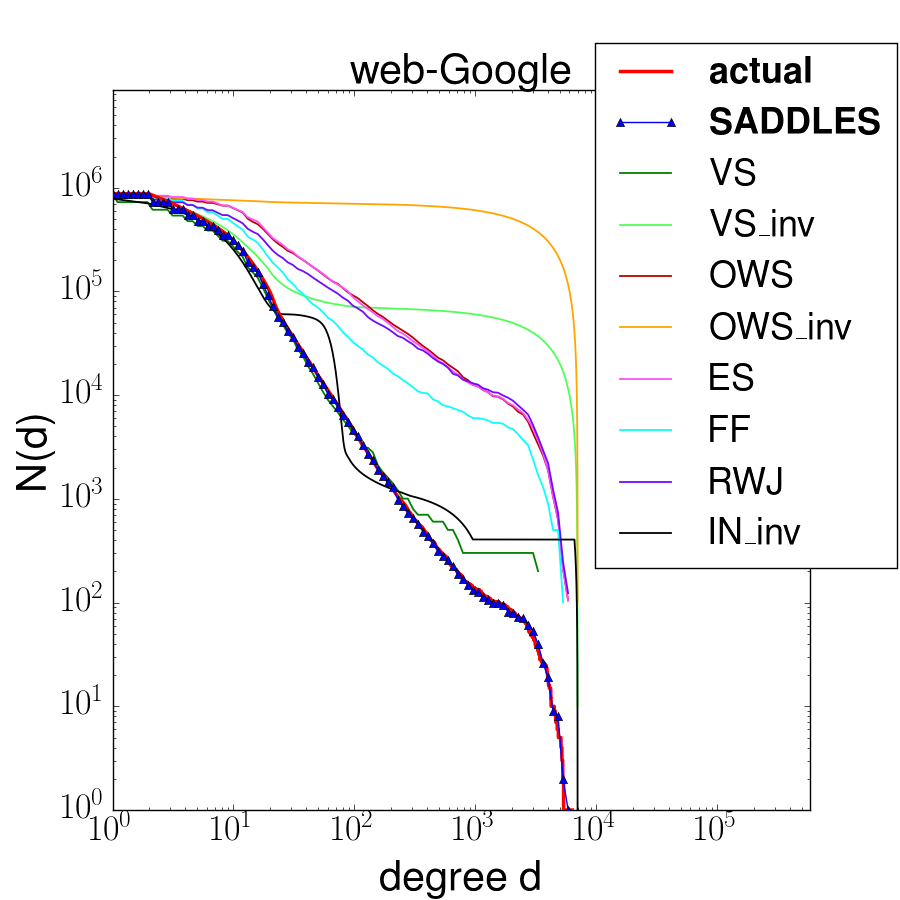}
    \label{fig:intro-google}}
\subfigure[{\tt cit-Patents} citation network]{\includegraphics[width=0.23\textwidth]{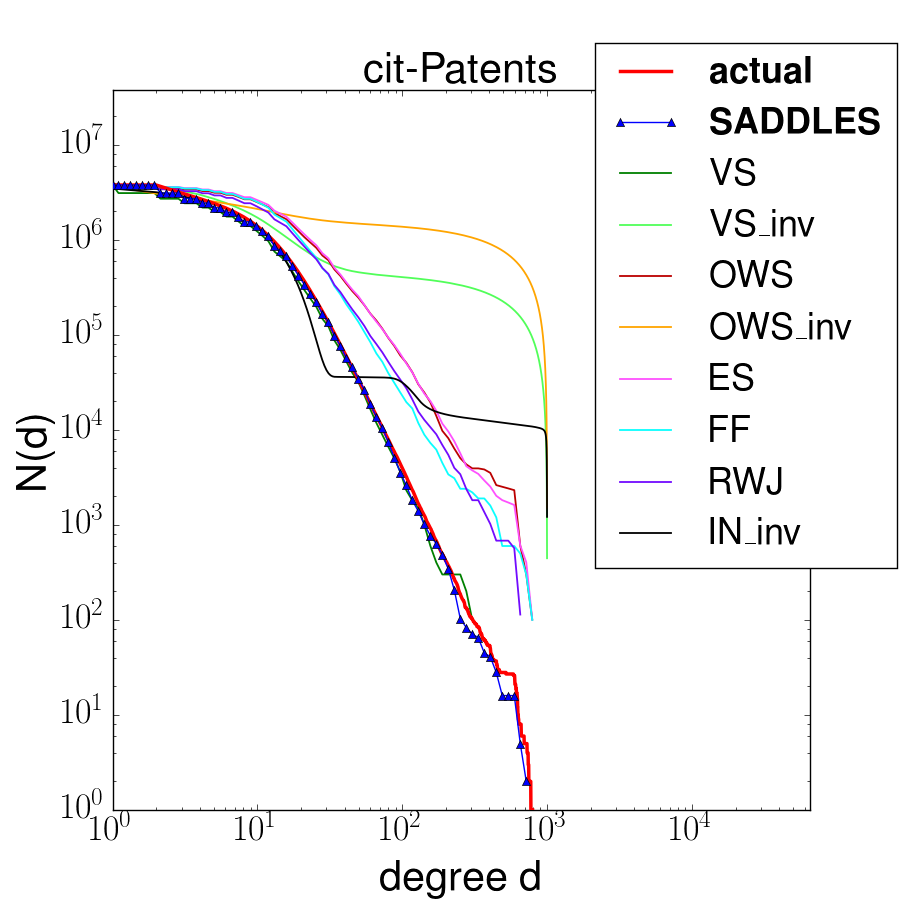}
    \label{fig:intro-patents}}
\subfigure[{\tt com-orkut} social network]{\includegraphics[width=0.23\textwidth]{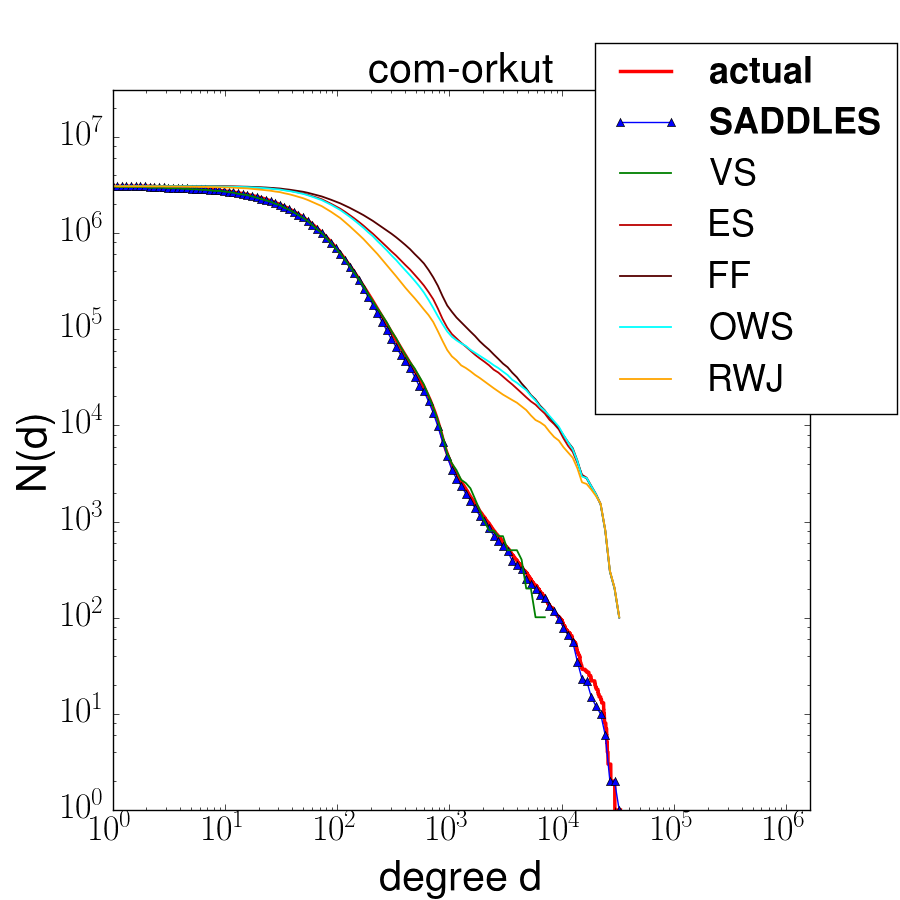}
    \label{fig:intro-orkut}}
\caption{The output of \mainalg{} on a collection of networks: {\tt amazon0601} (403K vertices, 4.9M edges), {\tt web-Google} (870K vertices, 4.3M edges), 
{\tt cit-Patents} (3.8M vertices, 16M edges), {\tt com-orkut} social network (3M vertices, 117M edges). \mainalg{}
samples $1\%$ of the vertices and gives accurate results for the entire (cumulative) degree distribution. For comparison,
we show the output of a number of sampling algorithms from past work, each run with the same number of samples. (Because of the size of {\tt com-Orkut},
methods involving optimization~\cite{ZKS15} fail to produce an estimate in reasonable time.)
}
\label{fig:intro}
\end{figure*}

There is a rich literature in statistics, data mining, and physics
on estimating graph properties (especially the degree distribution) using a small subsample~\cite{SW05, LF06, LKJ06, EHR08, ANK10, MB11, RT12, ahmed12socialnets, ANK14, ZKS15}.
Nonetheless, there is no provable algorithm for the {entire} degree distribution, with a formal analysis 
on when it is sublinear in the number of vertices. Furthermore, most empirical studies typically sample 10-30\% of the vertices
for reasonable estimates. 

\subsection{Problem description} \label{sec:problem}

We focus on
the \emph{complementary cumulative degree histogram}
(often called the cumulative degree distribution) or \emph{ccdh} of $G$.
This is the sequence $\{N(d)\}$,
where $N(d) = \sum_{r \geq d} n(r)$ is the number of vertices of degree at least $d$. The ccdh is typically used for fitting
distributions, since it averages out noise and is monotonic~\cite{ClShNe09}.
Our aim is to get an accurate bicriteria approximation to the ccdh of $G$, at all values
of $d$. 

\begin{definition} \label{def:approx} The sequence $\{\hN(d)\}$ is an $(\eps,\eps)$-estimate
of the ccdh if $\forall d$, $(1-\eps)N((1+\eps)d) \leq \hN(d) \leq (1+\eps)N((1-\eps)d)$.
\end{definition}

Computing an $(\eps,\eps)$-estimate is significantly harder than approximating
the ccdh using standard distribution measures. Statistical measures, such as the
KS-distance, $\chi^2$, $\ell_p$-norms, etc. tend to ignore the tail, since (in terms
of probability mass) it is a negligible portion of the distribution.
An $(\eps,\eps)$-estimate is accurate
for all $d$.

\textbf{The query model:} A formal approach requires specifying a \emph{query model} for accessing $G$.
We look to the subfields of property testing and sublinear algorithms within theoretical
computer science for such models~\cite{GR-bound,GR08}.
Consider the following three kinds of \emph{queries}.
\begin{asparaitem}
    \item Vertex queries: acquire a uniform random vertex $v \in V$.
    \item Neighbor queries: given $v \in V$, acquire  a uniform random neighbor $u$ of $V$.
    \item Degree queries: given $v \in V$, acquire  the degree $d_v$.
\end{asparaitem}
An algorithm is only allowed to make these queries to process the input.
It has to make some number of queries, and finally produce an output.
We discuss two query models, and give results for both.

\medskip

\begin{asparaitem}
    \item [\bf The Standard Model (SM)] All queries allowed: This is the standard model in numerous sublinear algorithms results~\cite{GR-bound,GR08,GRS11,ELRS15,ERS17}.
    Furthermore, most papers on graph sampling implicitly use this model for generating subsamples. Indeed, any method
    involving crawling from a random set of vertices and collecting degrees is in the SM. This model
    is the primary setting for our work, and allows for comparison with rich body of graph sampling algorithms.
    It is worth noting that in the SM, one can determine the entire degree
distribution in $O(n\log n)$ queries (the extra $\log n$ factor comes from the coupon collector bound
of finding all the vertices through uniform sampling). Thus, it makes sense to express the number of queries made by an algorithm as a fraction
of $n$. Alternately, the number of queries is basically the number of vertices encountered by the algorithm.
Thus, a sublinear algorithm makes $o(n)$ queries.
    
    \item [\bf The Hidden Degrees Model (HDM)] Vertex and neighbor queries allowed, not degree queries: This is a substantially weaker model. In numerous cybersecurity and network monitoring settings, an algorithm cannot query for degrees, and has to infer them
    indirectly. Observe that this model is significantly harder than the SM. It takes $O((m+n)\log n)$ to determine all the degrees, since one has to at least 
visit all the edges to find degrees exactly. In this model, we express the number of queries as a fraction
of $m$. 

\end{asparaitem}

\paragraph{Regarding uniform random vertex queries:} This is a fairly powerful query, that may not be
realizable in all situations. Indeed, Chierichetti et al. explicitly study this problem in social networks
and design (non-trivial) algorithms for sampling uniform random vertices~\cite{ChDa+16}. In a previous work,
Dasgupta, Kumar, and Sarlos study algorithms for estimating average degree when only random walks are possible~\cite{DaKu14}.
Despite this power, we believe
that SM is a good testbed for understanding \emph{when} a small sample of a graph provably gives properties
of the whole. Furthermore, in the context of graph sampling, access to uniform 
random vertices is commonly (implicitly) assumed~\cite{RT12,LKJ06,EHR08,LF06,PSE15,ZKS15,ANK14}.
The vast majority of experiments conducted often use uniform random vertices.

As a future direction, we believe it is important to investigate sampling models without random vertex queries.

\subsection{Our contributions} \label{sec:contribute}

Our main theoretical result is a new sampling algorithm, the 
Sublinear Approximations
for Degree Distributions Leveraging Edge Samples, or \mainalg.
This algorithm provably provides $(\eps,\eps)$-approximations
for the ccdh. We show how to design \mainalg{} under both the SM
and the HDM.
We apply \mainalg{} on a variety of real datasets and demonstrate its ability to accurately
approximate the ccdh with a tiny sample of the graph.
\begin{asparaitem}
    \item \textbf{Sampling algorithm for estimating ccdh:} Our algorithm combines a number of techniques in random sampling
    to get $(\eps,\eps)$-estimates for the ccdh. A crucial component is an application of 
    an edge simulation technique, first devised by Eden et al. in the context of triangle counting~\cite{ELRS15,ERS17}.
    This (theoretical) technique shows how to get a collection of weakly correlated uniform random edges from
    independent uniform vertices. \mainalg{} employs a weighting scheme on top of this method to estimate the ccdh.
    \item \textbf{Heavy tails leads to sublinear algorithms:} The challenge in analyzing \mainalg{} is
    in finding parameters of the ccdh that allow for sublinear query complexity.
    To that end, we discuss
    two parameters that measure ``heaviness"
    of the distribution tail: the classic $\hindex$-index and a newly defined $\findex$-index. We prove
    that the query complexity of \mainalg{} is sublinear (for both models) whenever these indices are large.
    \item \textbf{Excellent empirical behavior:} We deploy an implementation of \mainalg{} on a collection
    of large real-world graphs. In all instances, we achieve extremely accurate estimates for the entire ccdh
    by sampling at most $1\%$ of the vertices of the graph.
    Refer to \Fig{intro}. Observe how \mainalg{} tracks various jumps in the ccdh, for all graphs
    in \Fig{intro}.
    \item \textbf{Comparison with existing sampling methods:} A number of graph sampling methods
    have been proposed in practice, such as vertex sampling (VS), snowball sampling (OWS), forest-fire sampling (FF), induced graph sampling (IN), random walk (RWJ), edge sampling (ES)
    ~\cite{RT12,LKJ06,EHR08,LF06,PSE15,ZKS15,ANK14}. A recent work of Zhang et al. explicitly addresses biases in these sampling methods,
    and fixes them using optimization techniques~\cite{ZKS15}. We run head-to-head comparisons with all these sampling methods,
    and demonstrate the \mainalg{} gives significantly better practical performance. \Fig{intro} shows the output
    of all these sampling methods with a total sample size of $1\%$ of the vertices. Observe how across the board,
    the methods make erroneous estimates for most of the degree distribution. The errors are also very large, for all
    the methods. This is consistent with previous work, where methods sample more than 10\% of the number of vertices.
\end{asparaitem}

\subsection{Theoretical results in detail} \label{sec:theoretical}

Our main theoretical result is a new sampling algorithm, the Sublinear Approximations
for Degree Distributions Leveraging Edge Samples, or \mainalg.

We first demonstrate our results for power law degree distributions~\cite{BarabasiAlbert99,BrKu+00,FFF99}.
Statistical fitting procedures suggest they occur to some extent in the real-world, albeit with much noise~\cite{ClShNe09}.
The classic power law degree distribution sets $n(d) \propto 1/d^\gamma$, where $\gamma$ is typically in $[2,3]$.
We build on this to define a power law lower bound.

\begin{definition} \label{def:powerlaw} Fix $\gamma > 2$. A degree distribution
is bounded below by a power law with exponent $\gamma$, if the ccdh satisfies
the following property. There exists a constant $\tau > 0$ such
that for all $d$, $N(d) \geq \lfloor \tau n/d^{\gamma-1} \rfloor$.
\end{definition}

The following is a corollary of our main result.
For convenience, we will suppress query complexity dependencies on $\eps$ and $\log n$ factors,
using $\otilde(\cdot)$.

\begin{theorem} \label{thm:power} 
Suppose the degree distribution of $G$
is bounded below by a power law with exponent $\gamma$. Let the average degree be denoted by $\overline{d}$.
For any $\eps > 0$, the \mainalg{} algorithm outputs (with high probability) an $(\eps,\eps)$-approximation to the ccdh
and makes the following number of queries.
\begin{asparaitem}
    \item SM: $\otilde(n^{1-\frac{1}{\gamma}} + n^{1-\frac{1}{\gamma-1}} \overline{d})$
    \item HDM: $\otilde(n^{1-\frac{1}{2(\gamma-1)}} \overline{d})$
\end{asparaitem}
\end{theorem} 

In most real-world instances, the average degree $\overline{d}$ is typically constant. Thus, the complexities above are strongly sublinear.
For example, when $\gamma = 2$, we get $\otilde(n^{1/2})$
for both models. When $\gamma = 3$, we get $\otilde(n^{2/3})$
and $\otilde(n^{3/4})$.

Our main result is more nuanced, and holds for all degree distributions. 
If the ccdh has a heavy tail, we expect
$N(d)$ to be reasonably large even for large values of $d$.
We describe two formalisms of this notion, through
\emph{fatness indices}.

\begin{definition} The \emph{$h$-index} of the degree distribution is the largest $d$
such that there are at least $d$ vertices of degree at least $d$. 
\end{definition}

This is the exact analogy of the bibliometric $h$-index~\cite{H05}.
As we show in the \Sec{indices},
$h$ can be approximated by $\min_d (d+N(d))/2$. 
A more stringent index is obtained by replacing the arithmetic mean by 
the (smaller) geometric mean.

\begin{definition} \label{def:findex} The \emph{$z$-index} of the degree
distribution is $z = \min_{d: N(d) > 0} \sqrt{d\cdot N(d)}$.
\end{definition}

Our main theorem asserts that large $h$ and $z$ indices lead to a sublinear algorithm
for degree distribution estimation. \Thm{power} is a direct corollary
    obtained by plugging in values of the indices for power
    laws.

\begin{theorem} \label{thm:main-short} 
For any $\eps > 0$, the \mainalg{} algorithm outputs (with high probability) an $(\eps,\eps)$-approximation to the ccdh,
and makes the following number of queries.
\begin{asparaitem}
    \item SM: $\otilde(n/h + m/z^2)$
    \item HDM: $\otilde(m/z)$
\end{asparaitem}
\end{theorem}

\subsection{Challenges and Main Idea} \label{sec:challenges}

The heavy-tailed behavior of the real degree distribution poses the primary challenge to computing
$(\eps,\eps)$-estimates to the ccdh. As $d$ increases, there are fewer and fewer vertices of that degree.
Sampling uniform random vertices is inefficient when $N(d)$ is small.
A natural idea to find high degree vertices to pick a random neighbor of a random vertex. Such a sample
is more likely to be a high degree vertex. This is the idea behind methods like snowball sampling, forest fire sampling, random walk sampling,
graph sample-and-hold, etc.~\cite{RT12,LKJ06,EHR08,LF06,PSE15,ZKS15,ANK14}. But these lead to biased samples, since
vertices with the same degree may be picked with differing probabilities.

%
A direct extrapolation/scaling of the degrees in the observed graph does not provide an accurate
estimate. Our experiments show that existing methods always miss the head or the tail.
A more principled approach was proposed recently by Zhang et al.~\cite{ZKS15}, by casting
the estimation of the unseen portion of the distribution as an optimization problem. 
From a mathematical standpoint, the vast majority of existing results tend to analyze the KS-statistic,
or some $\ell_p$-norm. As we mentioned earlier, this does not work well for measuring the quality
of the estimate at all scales. As shown by our experiments, none of these methods give accurate
estimate for the entire ccdh with less than $5\%$ of the vertices.

The main innovation in \mainalg{} comes through the use of a recent theoretical
technique to simulate edge samples through vertex samples~\cite{ELRS15,ERS17}. 
The sampling of edges occurs through two stages. In the first stage,
the algorithm samples a set of $r$ vertices and sets up a distribution over 
the sampled vertices such that any edge adjacent to a sampled vertex may be sampled with uniform probability. In the second stage, it samples $q$ edges
from this distribution. While a single edge is uniform random, the set of edges are correlated. 

For a given $d$, we define a weight function on the edges, such that the total weight is exactly $N(d)$.
\mainalg{} estimates the total weight by scaling up the average weight on a random
sample of edges, generated as discussed above. The difficulty in the analysis is the correlation
between the edges. Our main insight is that if the degree distribution has a fat tail,
this correlation can be contained even for sublinear $r$ and $q$. Formally, this is achieved
by relating the concentration behavior of the average weight of the sample to
the $h$ and $z$-indices. The final algorithm combines this idea with vertex sampling
to get accurate estimates for all $d$.
 
The hidden degrees model is dealt with 
using birthday paradox techniques formalized by Ron and Tsur~\cite{RT16}. It is possible to estimate the degree $d_v$
using $O(\sqrt{d_v})$ neighbor queries. But this adds overhead to the algorithm,
especially for estimating the ccdh at the tail. As discussed earlier, we need methods that bias
towards higher degrees, but this significantly adds to the query cost of actually estimating
the degrees. 

\subsection{Related Work} \label{sec:related}

There is a rich body of literature on generating a graph sample that reveals graph properties of the larger
``true" graph. We do not attempt to fully survey this literature, and only refer to results
directly related to our work. The works of Leskovec \& Faloutsos~\cite{LF06},
Maiya \& Berger-Wolf~\cite{MB11}, and Ahmed, Neville, \& Kompella~\cite{ANK10,ANK14} provide excellent surveys of multiple sampling methods.

There are a number of sampling methods based on random crawls: forest-fire~\cite{LF06},
snowball sampling~\cite{MB11}, and expansion sampling~\cite{LF06}. As has been detailed
in previous work, these methods tend to bias certain parts of the network, which can
be exploited for more accurate estimates of various properties~\cite{LF06, MB11, RT12}.
A series of papers by Ahmed, Neville, and Kompella~\cite{ANK10,ahmed12socialnets,ANK14,ahmed2014graph}
have proposed alternate sampling methods that combine random vertices and edges
to get better representative samples. 

All these results aim to capture numerous properties of the graph, using a single graph sample.
Nonetheless, there is much previous work focused on the degree distribution.
Ribiero and Towsley~\cite{RT12}
and Stumpf and Wiuf~\cite{SW05} specifically study degree distributions. Ribiero and Towsley~\cite{RT12} do detailed analysis
on degree distribution estimates (they also look at the ccdh) for a variety of these sampling
methods. Their empirical results show significant errors either at the head or the tail.
We note that almost all these results end up sampling up to 20\% of the graph to estimate the degree distribution.

Zhang et al. observe that the degree distribution of numerous sampling methods is a random
linear projection of the true distribution~\cite{ZKS15}. They attempt to invert this (ill-conditioned)
linear problem, to correct the biases. This leads to improvement in the estimate, but the empirical
studies typically sample more than $10\%$ of the vertices for good estimates.

A recent line of work by Soundarajan et al. on \emph{active probing} also has flavors of graph sampling~\cite{SEGP16,SEGP17}.
In this setting, we start with 
a small, arbitrary subgraph and try to grow this subgraph to achieve some coverage objective
(like discover the maximum new vertices, find new edges, etc.). 
The probing schemes devised in these papers outperform uniform random sampling methods
for coverage objectives.

Some methods try to match the shape/family of the distribution, rather than estimate it as a whole~\cite{SW05}.
Thus, statistical methods can be used to estimate parameters of the distribution.
But it is reasonably well-established that real-world degree distributions are rarely pure
power laws in most instances~\cite{ClShNe09}. Indeed, fitting a power law is rather challenging and naive regression
fits on log-log plots are erroneous, as results of Clauset-Shalizi-Newman
showed ~\cite{ClShNe09}.
%

The subfield of \emph{property testing and sublinear algorithms for sparse graphs} within theoretical computer science
can be thought of as a formalization of graph sampling to estimate properties.
Indeed, our description of the main problem follows this language.
There is a very rich body of mathematical work in this area (refer to Ron's survey~\cite{Ron-survey}).
Practical applications of graph property testing are quite rare, and we are only
aware of one previous work on applications for finding dense cores in router networks~\cite{GRWW08}.
The specific problem of estimating the average degree (or the total
number of edges) was studied by Feige~\cite{F06} and Goldreich-Ron~\cite{GR08}.
Gonen et al. and Eden et al. focus on the problem of estimating higher moments 
of the degree distribution~\cite{GRS11,ERS17}. One of the main techniques we use
of simulating edge queries was developed in sublinear algorithms results of Eden et al.~\cite{ELRS15, ERS17}
in the context of triangle counting and degree moment estimation. We stress that all
these results are purely theoretical, and their practicality is by no means obvious.

On the practical side, Dasgupta, Kumar, and Sarlos study average degree estimation in 
real graphs, and develop alternate algorithms~\cite{DaKu14}. They require the graph to have low mixing time
and demonstrate that the algorithm has excellent behavior in practice (compared to implementations
of Feige's and the Goldreich-Ron algorithm~\cite{F06,GR08}). Dasgupta et al. note
that sampling uniform random vertices is not possible in many settings, and thus
they consider a significantly weaker setting than SM or HDM. Chierichetti et al.
focus on  sampling uniform random vertices, using only a small set of
seed vertices and neighbor queries~\cite{ChDa+16}.

We note that there is a large body of work on sampling graphs from a stream~\cite{mcgregor2014graph}.
This is quite different from our setting, since a streaming algorithm
observes every edge at least once. The specific problem of estimating the 
degree distribution at all scales was considered by Simpson et al.~\cite{OSM15}.
They observe many of the challenges we mentioned earlier: the difficulty of estimating
the tail accurately, finding vertices at all degree scales, and combining estimates from the head and the tail.

\section{Preliminaries} \label{sec:prelim}
We say that the input graph $G$ has $n$ vertices and $m$ edges and $m \geq n$ (since isolated vertices are not relevant here).
For any vertex $v$, let $\Gamma(v)$ be the neighborhood of $v$, and $d_v$ be the degree. As mentioned earlier, $n(d)$ is the number of vertices of degree $d$
and $N(d) = \sum_{r \geq d} n(r)$ is the ccdh at $d$. We use ``u.a.r."
as a shorthand for ``uniform at random". We stress that the all mention of probability and error is with respect to
the randomness of the sampling algorithm. There is no stochastic assumption on the input graph $G$. We use the shorthand $A \in (1\pm\alpha)B$ for $A \in [(1-\alpha)B, (1+\alpha)B]$. We will apply the following (rescaled) Chernoff bound.

\begin{theorem} \label{thm:chernoff} [Theorem 1 in~\cite{DuPa}] Let $X_1, X_2, \ldots, X_k$
be a sequence of iid random variables with expectation $\mu$. Furthermore, $X_i \in [0,B]$. 
\begin{asparaitem}
    \item For $\eps< 1$, $ \Pr[|\sum_{i=1}^k X_i - \mu k| \geq \eps\mu k] \leq 2\exp(-\eps^2 \mu k/3B) $.
    \item For $t \geq 2e\mu$, $\Pr[\sum_{i=1}^k X_i \geq tk] \leq 2^{-tk/B}$.
\end{asparaitem}
\end{theorem}

%
%

\subsection{More on Fatness indices} \label{sec:indices} 

The following
characterization of the $h$-index will be useful for analysis.
Since $(d+N(d))/2 \leq \max(d,N(d)) \leq d+N(d)$, this
proves that $\min_d (d+N(d))/2$ is a $2$-factor approximation
to the $h$-index.

\begin{lemma} \label{lem:hindex} $\min_d \max(d,N(d)) \in \{h,h+1\}$
\end{lemma}

\begin{proof} Let $s = \min_d \max(d,N(d))$ and let the minimum be attained at $d^*$.
If there are multiple minima, let $d^*$ be the largest among them. We consider two cases.
(Note that $N(d)$ is a monotonically non-increasing sequence.)

Case 1: $N(d^*) \geq d^*$. So $s = N(d^*)$. Since $d^*$ is the largest minimum, for any $d > d^*$,
$d > N(d^*)$. (If not, then the minimum is also attained at $d > d^*$.) Thus, $d > N(d^*) \geq N(d)$.
For any $d < d^*$, $N(d) \geq N(d^*) \geq d^* > d$. We conclude that $d^*$ is largest $d$
such that $N(d) \geq d$. Thus, $h = d^*$.

If $s \neq h$, then $d^* < N(d^*)$. Then, $N(d^*+1) < N(d^*)$, otherwise the minimum would
be attained at $d^*+1$. Furthermore, $\max(d^*+1,N(d^*+1)) > N(d^*)$, implying $d^*+1 > N(d^*)$.
This proves that $h+1 > s$.

Case 2: $d^* > N(d^*)$. So $s = d^*$. For $d > d^*$, $N(d) \leq N(d^*) < d^* < d$. 
For $d < d^*$, $N(d) \geq d^* > d$ (if $N(d) < d^*$, then $d^*$ would not be the minimizer).
Thus, $d^* - 1$ is the largest $d$ such that $N(d) \geq d$, and $h = d^*-1 = s-1$.
\end{proof}

The $h$-index does not measure $d$ vs $N(d)$ at different scales,
and a large $h$-index only ensures that there are ``enough'' high-degree vertices.
For instance, the h-index   does not distinguish between two different distributions whose ccdh $N_1$ and $N_2$  are such that  $N_1(100)=100$ and $N_1(d)=0$ for $d>100$, and $N_2(100,000)=100$ and $N_2(d)=100$ for all other values of $d\geq100$. The $h$-index in both these cases is 100.

The $h$ and $z$-indices are related to each other.

\begin{claim} \label{clm:index} $\sqrt{h} \leq z \leq h$.
\end{claim}

\begin{proof} Since $N(d)$ is integral, if $N(d) > 0$, then $N(d) \geq 1$.
Thus, for all $N(d) > 0$,  $\sqrt{\max(d,N(d))} \leq d\cdot N(d) \leq \max(d,N(d))$.
We take the minimum over all $d$ to complete the proof.
\end{proof}

To give some intuition about these indices, we compute the $h$ and $z$ index for power laws.
The classic power law degree distribution sets $n(d) \propto 1/d^\gamma$, where $\gamma$ is typically in $[2,3]$.

\begin{claim} \label{clm:power} If a degree distribution is bounded below by
a power law with exponent $\gamma$, then $h = \Omega(n^\frac{1}{\gamma})$
and $z = \Omega(n^\frac{1}{2(\gamma-1)})$.
\end{claim} 

\begin{proof} Consider $d \leq \tau n^{1/\gamma}$, where $\tau$ is defined
according to \Def{powerlaw}. Then, $N(d) \geq \lfloor \tau n/(\tau^{1/\gamma} n^{(\gamma-1)/\gamma}) \rfloor$
$= \Omega(n^{1/\gamma})$. This proves the $h$-index bound.

Set $d^* = (\tau n)^{\frac{1}{\gamma-1}}$.
For $d \leq d^*$, $N(d) \geq 1$ and 
$d\cdot N(d) \geq (\tau/2) n/d^{\gamma-2}$ 
$= \Omega(n^\frac{1}{\gamma-1})$.
If there exists no $d > d^*$ such that $N(d) > 0$,
then $z = \Omega(n^\frac{1}{2(\gamma-1)})$.
If there does exist some such $d$, then $z = \Omega(\sqrt{d^*})$
which yields the same value.
\end{proof}

Plugging in values, for $\gamma = 2$, both $h$ and $z$ are $\Omega(\sqrt{n})$.
For $\gamma = 3$, $h = \Theta(n^{1/3})$ and $z = \Theta(n^{1/4})$.  

\subsection{Simulating degree queries for HDM} \label{sec:hdm}

The Hidden Degrees Model does not allow for querying the degree $d_v$ of a vertex $v$. Nonetheless,
it is possible to get accurate estimates of $d_v$ by sampling u.a.r. neighbors (with replacement)
of $v$. This can be done by using the birthday paradox argument,
as formalized by Ron and Tsur~\cite{RT16}.
Roughly speaking, one repeatedly samples neighbors until the same vertex is seen twice.
If this happens after $t$ samples, $t^2$ is a constant factor approximation for $d_v$.
This argument can be refined to get accurate approximations for $d_v$ using $O(\sqrt{d_v})$
random edge queries.

\begin{theorem} \label{thm:birthday} [Theorem 3.1 of~\cite{RT16}, restated]
Fix any $\alpha > 0$.  There is an algorithm that outputs a value in $(1\pm\alpha)d_v$
with probability $> 2/3$, and makes an expected $O(\sqrt{d_v}/\alpha^2)$ u.a.r. neighbor samples.
\end{theorem}

For the sake of the theoretical analysis, we will simply assume this theorem. In
the actual implementation of \mainalg{}, we will discuss the specific parameters
used. It will be helpful to abstract out the estimation of degrees through the following
corollary. The procedure \degest$(v)$ will be repeatedly invoked by \mainalg.
This is a direct consequence of setting $\alpha - \eps/10$ and applying \Thm{boost}
with $\delta = 1/n^3$.

\begin{corollary} \label{cor:birthday} There is an algorithm \degest{} that takes
as input a vertex $v$, and has the following properties:
\begin{asparaitem}
    \item For all $v$: with probability $> 1 - 1/n^3$,
    the output \degest$(v)$ is in $(1\pm \eps/10) d_v$.
    \item The expected running time and query complexity of \degest$(v)$ is $O(\eps^{-2}\sqrt{d_v}\log n)$.
\end{asparaitem}
\end{corollary}

%

We will assume that invocations to $\degest$ with the same arguments use the same sequence of random bits.
Alternately, imagine that a call to $\degest(v,\eps)$ stores the output, so subsequent calls output the same value. For the sake of analysis, it is convenient to imagine that \degest{$(v)$} is called once for all vertices $v$, and these results are stored.

\begin{definition} \label{def:deg} The output \degest$(v)$ is denoted by $\hat{d_v}$.
The random bits used in all calls to \degest{} is collectively denoted $\Lambda$.
(Thus, $\Lambda$ completely specifies all the values $\{\hat{d_v}\}$.)
We say $\Lambda$ is \emph{good} if $\forall v \in V$, $\hat{d_v} \in (1\pm \eps/10)d_v$.
\end{definition}

The following is a consequence of conditional probabilities.

\begin{claim} \label{clm:cond} Consider any event $\mathcal{A}$, such that for any good $\Lambda$,
$\Pr[\mathcal{A} | \Lambda] \geq p$. Then $\Pr[\mathcal{A}] \geq p - 1/n^2$.
\end{claim} 

\begin{proof} The probability that $\Lambda$ is not good is at most
the probability that for \emph{some} $v$, \degest$(v) \notin (1\pm\eps/10)$.
By the union bound and \Cor{birthday}, the probability is at most $1/n^2$.

Note that \\ $\Pr[\mathcal{A}] \geq \sum_{\Lambda \textrm{good}} \Pr[\Lambda] \Pr[\mathcal{A} | \Lambda]$
$ \geq p \Pr[\Lambda \textrm{ is good}]$. 
Since $\Lambda$ is good with probability at least $1-1/n^2$, $\Pr[\mathcal{A}] \geq (1-1/n^2) p \geq p - 1/n^2$.
\end{proof}

For any fixed $\Lambda$, we set $\widehat{N_\Lambda}(d)$ to be $|\{v | \hat{d_v} \geq d\}|$.
We will perform the analysis of \mainalg{} with respect to the $\widehat{N_\Lambda}$-values.

\begin{claim} \label{clm:lambda} Suppose $\Lambda$ is good. For all $v$,
$\widehat{N_\Lambda}(v) \in [N((1+\eps/9)d), N((1-\eps/9)d)]$.
\end{claim}

\begin{proof} Since $\Lambda$ is good, $\forall u, \hat{d_u} \in (1\pm \eps/10) d_u$,
Furthermore, if $d_u \geq (1+\eps/9)d$, then $\hat{d_u} \geq (1-\eps/10)(1+\eps/9)d \geq d$.
Analogously, if $d_u \leq (1-\eps/9)d$, then $\hat{d_u} \leq (1+\eps/10)(1-\eps/9)d \leq d$.
Thus, 
$\{u| d_u \geq d(1+\eps/9)\} \subseteq \{u| \hat{d_u} \geq d\}
\subseteq \{u | d_u \geq d(1-\eps/9)\}$.
\end{proof}

\section{The Main Result and \mainalg{}}

We begin by stating the main result, and explaining how heavy tails lead to sublinear algorithms. Note that $D$ refers to a set of degrees, for which we desire an approximation to $N(d)$.

\begin{theorem} \label{thm:main} 
There exists an algorithm \mainalg{} with the following properties.
Let $c$ be a sufficiently large constant. 
Fix any $\eps > 0, \err > 0$. Suppose that the parameters of \mainalg{}
satisfy the following conditions: $r \geq c\eps^{-2}n/\hindex$, $q \geq c\eps^{-2}m/\findex^2$,
$\rep \geq c\log(n/\err)$, $\thresh \geq c\eps^{-2}$.

Then with probability at least $1-\delta$, for all $d \in D$, \mainalg{} outputs an $(\eps,\eps)$-approximation of $N(d)$.

The expected number of queries made depends on the model, \emph{and is independent of the size of $D$}.
\begin{asparaitem}
    \item SM: $O((n/\hindex + m/\findex^2)(\eps^{-2}\log(n/\err)))$.
    \item HDM: $O((m/\findex)(\eps^{-4}\log^2(n/\err)))$.
\end{asparaitem}
\end{theorem}

Observe how a larger $h$ and $z$-index lead to smaller running times.
Ignoring constant factors and assuming $m = O(n)$, asymptotically
increasing $h$ and $z$-indices lead to sublinear algorithms.

We now describe the algorithm itself.
The main innovation in \mainalg{} comes through the use of a recent theoretical
technique to simulate edge samples through vertex samples~\cite{ELRS15,ERS17}. 
The sampling of edges occurs through two stages. In the first stage,
the algorithm samples a set of $r$ vertices and sets up a distribution over 
the sampled vertices such that any edge adjacent to a sampled vertex may be sampled with uniform probability. In the second stage, it samples $q$ edges
from this distribution.

For each edge, we compute a weight based on the degrees of its vertices and generate our estimate
by averaging these weights. Additionally, we use vertex sampling to estimate the head
of the distribution. Straightforward Chernoff bound arguments can be used to determine
when to use the vertex sampling over the edge sampling method.

The same algorithmic structure is used for the Standard Model and the Hidden Degrees Model.
The only difference is the use the algorithm of \Cor{birthday} to estimate
degrees in the HDM, while the degrees are directly available in the Standard Model.


%
%
%
%

\begin{algorithm} 
\caption{\mainalg$(D,r,q,\rep,\thresh)$ \label{alg:SADDLES} \newline 
{\bf Inputs:} \newline
\textit{$D$}: set of degrees for which N(d) is to be computed \newline
\textit{$r$}: budget for vertex samples \newline
\textit{$q$}: budget for edge samples \newline
\textit{$\rep$}: boosting parameter \newline
\textit{$\thresh$}: cutoff for vertex sampling \newline
{\bf Output:} \newline
$\{N^\prime(d)\}$: estimated $\{N(d)\}$
} \label{alg:edgest}
For $t=1, \ldots, \rep$: \\ \label{step:boost}
\ \ \ \ For $i=1, \ldots, r$: \\
\ \ \ \ \ \ \ \ Select u.a.r. vertex $v$ and add it to multiset $R$. \label{step:defS}\\
\ \ \ \ \ \ \ \ In HDM, call \degest$(v)$ to get $\hat{d_v}$. In SM, set $\hat{d_v}$ to $d_v$.\label{step:deg1}\\
\ \ \ \ \ \ \ \ For $d \in D$: \\
\ \ \ \ \ \ \ \ \ \ \ \ If $\hat{d_v} \geq d$, set $X_{id} = 1$. Else, $X_{id} = 0$. \\
\ \ \ \ Let $\hat{d}_R = \sum_{v \in R} \hat{d_v}$ and $\cD$ denote the distribution over $R$ where $v \in R$
is selected with probability $\hat{d_v}/\hat{d}_R$. \\
\ \ \ \ For $i=1, \ldots, q$: \label{step:set_q} \\
\ \ \ \ \ \ \ \ Sample $v \sim \cD$. \\
\ \ \ \ \ \ \ \ Pick u.a.r. neighbor $u$ of $v$. \\
\ \ \ \ \ \ \ \ In HDM, call \degest$(u)$ to get $\hat{d_u}$. In SM, set $\hat{d_u}$ to $d_u$. \label{step:deg2}\\
\ \ \ \ \ \ \ \ For $d \in D$: \\
\ \ \ \ \ \ \ \ \ \ \ \ If $\hat{d_u} \geq d$, set $Y_{id} = 1/\hat{d_u}$. Else, set $Y_{id} = 0$. \label{step:Yi} \\
\ \ \ \ For $d \in D$: \\
\ \ \ \ \ \ \ \ If $\sum_{i \leq r} X_{id} \geq \thresh$: \label{step:thresh} \\
\ \ \ \ \ \ \ \ \ \ \ \ $\hNd[t] = \frac{n}{r} \sum_{i \leq r} X_{id}$. \label{step:vertex} \\
\ \ \ \ \ \ \ \ else $\hNd[t]=\frac{n}{r} \cdot \frac{\hat{d}_R}{q}\sum_{i \leq q} Y_{id}$. \label{step:el_2} \label{step:def_El} \\
For $d \in D$:  \\ 
\ \ \ \ $N^\prime(d) = median\{\hNd\}$ \\
Return $\{N^\prime(d)\}$
\end{algorithm}
\paragraph{The core theoretical bound:} The central technical bound
deals with the properties of each individual estimate $\hNd[t]$.

\begin{theorem} \label{thm:mainbound} Suppose $r \geq c\eps^{-2}n/\hindex$, $q \geq c\eps^{-2} m/\findex^2$, $\thresh = c\eps^{-2}$.
Then, for all $d \in D$, with probability $\geq 5/6$, $\hNd[t] \in [(1-\eps/2)N((1+\eps/2)d), (1+\eps/2)N((1-\eps/2)d]$.
\end{theorem}

The proof of this theorem is the main part of our analysis, which appears in the next section.
\Thm{main} can be derived from this theorem, as we show next.

\begin{proof} (of \Thm{main}) First, let us prove the error/accuracy bound.
For a fixed $d \in D$ and $t \leq \ell$, \Thm{mainbound}
asserts that we get an accurate estimate with probability $\geq 5/6$.
Among the $\ell$ independent invocations, the probability that more
than $\ell/3$ values of $\hNd[t]$ lie outside $[(1-\eps/2)N((1+\eps/2)d), (1+\eps/2)N((1-\eps/2)d]$
is at most $\exp(-\ell/100)$ (by the Chernoff bound of \Thm{chernoff}).
By the choice of $\ell \geq c\log(n/\delta)$, the probability is at most $\delta/n$.
Thus, with probability $>1-\delta/n$, the median of $\hNd[t]$ gives an $(\eps,\eps)$
estimate of $N(d)$. By a union bound over all (at most $n$) $d \in D$,
the total probability of error over any $d$ is at most $\delta$.

Now for the query complexity.
The overall algorithm is the same for both models, involving multiple invocations of \mainalg.
The only difference is in \degest, which is trivial when degree queries are allowed.
For the Standard Model, the number of graph queries made for a single invocation of \mainalg{} 
is simply $O(\ell(r+q))$ $=O(\eps^{-2}(n/h + m/z^2)\log (n/\delta))$.
%

For the Hidden Degrees Model,
we have to account for the overhead of \Cor{birthday} for each degree estimated.
The number of queries for a single call to \degest$(d)$ is $O(\eps^{-2}\sqrt{d_v}\log n)$.
The total overhead of all calls in \Step{deg1} is $\EX[\sum_{v \in R} \sqrt{d_v} (\eps^{-2}\log n)]$.
By linearity of expectation, this is $O((\eps^{-2}\log n)r\EX[\sqrt{d_v}]$, where the expectation
is over a uniform random vertex. We can bound $r\EX[\sqrt{d_v}]$ $\leq r\EX[d_v]$ $=O(\eps^{-2}n(m/n)/h)$
$= O(\eps^{-2}n/h)$. 

The total overhead of all calls in \Step{deg2} requires more care. Note that when \degest$(v)$
is called multiple times for a fixed $v$, the subsequent calls require no further queries.
(This is because the output of the first call can be stored.) We partition the vertices
into two sets $S_0 = \{v | d_v \leq \findex^2\}$ and $S_1 = \{v | d_v > \findex^2\}$.
The total query cost of queries to $S_0$ is at most $O(q\findex) = O((\eps^{-2}\log n)m/\findex)$.
For the total cost to $S_1$, we directly bound by (ignoring the $\eps^{-2}\log n$ factor)
$\sum_{v \in S_1} \sqrt{d_v} = \sum_{v \in S_1} d_v/\sqrt{d_v}$ $\leq \findex^{-1} \sum_v d_v $ $= O(m/\findex)$.
All in all, the total query complexity is $O((\eps^{-4}\log^2n) (n/h + m/\findex))$.
Since $m \geq n$ and $\findex \leq h$, we can simplify to $O((\eps^{-4}\log^2n)(m/\findex))$.
\end{proof}

%
%
%
%

\section{Analysis of \mainalg}

We now prove \Thm{mainbound}. There are a number of intermediate claims towards that.
We will fix $d \in D$ and a choice of $t$. Abusing notation, we use $\hNd$ to refer to $\hNd[t]$.
The estimate of \Step{vertex}
can be analyzed with a direct Chernoff bound.

\begin{claim} \label{clm:vertex} The following holds with probability $> 9/10$.
If \mainalg$(r,q)$ outputs an estimate in \Step{vertex} for a given $d$,
then $\hNd \in (1 \pm \eps/10) \aN(d)$. If it does not output in \Step{vertex},
then $\aN(d) < (2c/\eps^2)(n/r)$.
\end{claim}

\begin{proof} Each $X_i$ is an iid Bernoulli random variable, with success probability precisely $\aN(d)/n$.
We split into two cases.

Case 1: $\aN(d) \geq (c/10\eps^2)(n/r)$. By the Chernoff bound of \Thm{chernoff},
$\Pr[|\sum_{i \leq r} X_i - r \aN(d)/n| \geq (\eps/10)(r \aN(d)/n)]$
$\leq 2\exp(-(\eps^2/100) (r\aN(d)/n)$ $\leq 1/100$.

Case 2: $\aN(d) \leq (c/10\eps^2)(n/r)$. Note that $\EX[\sum_{i \leq r} X_i] \leq c/10\eps^2$ $\leq (c /\eps^2)/2e$.
By the upper tail bound of \Thm{chernoff}, $\Pr[\sum_{i \leq r} X_i \geq c/\eps^2] < 1/100$.

Thus, with probability at least $99/100$, if an estimate is output in \Step{vertex},
$\aN(d) > (c/10\eps^2)(n/r)$. By the first case, with probability at least $99/100$,
$\hNd$ is a $(1+\eps/10)$-estimate for $\aN(d)$. 
A union bound completes the first part.

Furthermore, if $\aN(d) \geq (2c/\eps^2)(n/r)$,
then with probability at least $99/100$, $\sum_{i \leq r} X_i \geq (1-\eps/10) r \aN(d)/n$ $\geq c/\eps^2 = \thresh$.
A union bound proves (the contrapositive of) the second part.
\end{proof}

We define weights of \emph{ordered} edges. The weight only depends
on the second member in the pair, but allows for a more convenient analysis.
The weight of $\langle v, u\rangle$ is the random variable $Y_i$ of \Step{Yi}.

\begin{definition} \label{def:weights}
The $d$-weight of an ordered edge $\langle v,u \rangle$ for a given $\Lambda$ (the randomness of \degest)
is defined as follows. We set $\wgt_{\Lambda,d}(\langle v,u\rangle)$ to be $1/\hat{d_u}$
if $\hat{d_u} \geq d$, and zero otherwise. For vertex $v$, $\wgt_{\Lambda,d}(v) = \sum_{u \in \Gamma(v)} \wgt_{\Lambda,d}(\langle v,u\rangle)$.
\end{definition}

The utility of the weight definition is captured by the following claim. The total weight is an approximation of $\hNd$, and thus, we can analyze how well \mainalg{} approximates the total weight.

\begin{claim} \label{clm:wt} If $\Lambda$ is good, $\sum_{v \in V} \wgt_{\Lambda,d}(v) \in (1\pm\eps/9) \widehat{N_\Lambda}(d)$.
\end{claim}

\begin{proof}
\begin{eqnarray} 
\sum_{v \in V} \wgt_{\Lambda,d}(v) & = & \sum_{v \in V} \sum_{u \in \Gamma(v)} \bone_{\hat{d_u} \geq d}/\hat{d_u} \nonumber \\
& = & \sum_{u: \hat{d_u} \geq d} \sum_{v \in \Gamma(u)} 1/\hat{d_u}
= \sum_{u: \hat{d_u} \geq d} d_u/\hat{d_u} \label{eqn:wgt}
\end{eqnarray}

Since $\Lambda$ is good, $\forall u, \hat{d_u} \in (1\pm \eps/10) d_u$,
and $d_u/\hat{d_u} \in (1 \pm \eps/9)$. 
Applying in \Eqn{wgt}, $\sum_{v \in V} \wgt_{\Lambda,d}(v)$
$\in (1\pm\eps/9) \widehat{N_\Lambda}(d)$.
\end{proof}

We come to an important lemma, that shows that the weight of the random subset $R$
(chosen in \Step{defS}) is well-concentrated. This is proven using a Chernoff bound,
but we need to bound the maximum possible weight to get a good bound on $r = |R|$.

\begin{lemma}\label{lem:wt} Fix any good $\Lambda$ and $d$. Suppose $r \geq c\eps^{-2}n/d$.
With probability at least $9/10$, 
$\sum_{v \in R}\wgt_{\Lambda,d}(v) \in (1\pm\eps/8) (r/n)\aN(d)$.
\end{lemma}

\begin{proof} Let $\wgt(R)$ denote $\sum_{v \in R} \wgt_{\Lambda,d}(v)$.
By linearity of expectation, $\EX[\wgt(R)] $ $= (r/n) \cdot $  $\sum_{v \in V} \wgt_{\Lambda,d}(v)$
$ \geq (r/2n) \aN(d)$. 
To apply the Chernoff bound, we need to bound the maximum
weight of a vertex. For good $\Lambda$, the weight $\wgt_{\Lambda,d}$ of
any ordered pair is at most $1/(1-\eps/10)d \leq 2/d$.
The number of neighbors of $v$ such that $\hat{d}_u \geq d$
is at most $\widehat{N_\Lambda}(d)$. Thus,
$\wgt_{\Lambda,d}(v) \leq 2 \widehat{N_\Lambda}(d)/d$.
 
By the Chernoff bound of \Thm{chernoff} and setting $r \geq c \eps^{-2}n/d$,
\begin{align*}
    & Pr\left[ \left\lvert \wgt(R) - \EX[\wgt(R)] \right\rvert > (\eps/20)\EX[\wgt(R)] \right] \\
    & < 2\exp\left(- \frac{\eps^2 \cdot (c \eps^{-2}n/d) \cdot (\aN(d)/2n) }{ 3 \cdot 20^2 \cdot 2\aN(d)/d}\right)  \leq 1/10
\end{align*} 
With probability at least $9/10$, $\wgt(R) \in (1\pm \eps/20)\EX[\wgt(R)]$.
By the arguments given above, $\EX[\wgt(R)] \in (1 \pm \eps/9) (r/n) \aN(d)$. We combine to complete the proof.
\end{proof}
   
Now, we determine the number of edge samples required to estimate the weight $\wgt_{\Lambda,d}(R)$.
    
\begin{lemma}\label{lem:edge}
    Let $\hNd$ be as defined in Step~\ref{step:el_2} of \  \mainalg. Assume $\Lambda$ is good, $r \geq c \eps^{-2}n/d$,
    and $q \geq c \eps^{-2}m/(d\aN(d))$. Then, with probability $> 7/8$, $\hNd \in (1\pm \eps/4)\aN(d)$.
\end{lemma}

\begin{proof} We define the random set $R$ selected in \Step{defS} to be \emph{sound}
if the following hold. (1) $\wgt(R) = \sum_{v \in R} \wgt_{\Lambda,d}(v)$
$\in (1 \pm \eps/8) (r/n) \aN(d)$ and (2) $\sum_{v \in R} d_v \leq 100 r(2m/n)$.
By \Lem{wt}, the first holds with probability $> 9/10$. Observe that
$\EX[\sum_{v \in R} d_v] = r(2m/n)$, since $2m/n$ is the average degree.
By the Markov bound, the second holds with probability $> 99/100$.
By the union bound, $R$ is sound with probability at least $1-(1/10 + 1/100) > 8/9$.

Fix a sound $R$. Recall $Y_i$ from \Step{Yi}. The expectation of $Y_i|R$
is $\sum_{v \in R} \Pr[\textrm{$v$ is selected}] \cdot $ $\sum_{u \in \Gamma(v)} \Pr[\textrm{$u$ is selected}] \wgt_{\Lambda,d}(\langle v,u\rangle)$.
We plug in the probability values, and observe that 
for good $\Lambda$, for all $v$, $\hat{d}_v/d_v \in (1\pm \eps/10)$.
\begin{eqnarray}
    \EX[Y_i | R] & = & \sum_{v \in R} (\hat{d}_v/\hat{d}_R) \sum_{u \in \Gamma(v)} (1/d_v) \wgt_{\Lambda,d}(\langle v,u\rangle)  \nonumber \\
    & = & (1/\hat{d}_R) \sum_{v \in R} (\hat{d}_v/d_v)\sum_{u \in \Gamma(v)} \wgt_{\Lambda, d}(\langle v,u\rangle) \nonumber \\
    & \in & (1 \pm \eps/10) (1/\hat{d}_R) \sum_{v \in R} \sum_{u \in \Gamma(v)} \wgt_{\Lambda, d}(\langle v,u\rangle) \nonumber \\
    & \in & (1 \pm \eps/10) (\wgt(R)/\hat{d}_R) \label{eqn:Yi-exp}
\end{eqnarray}
Note that $\hNd = (n/r) (\hat{d}_R/q) \sum_{i \leq q} Y_i$ and $(n/r) (\hat{d}_R/q) \EX[\sum_{i \leq q} Y_i|R]$
$\in (1\pm\eps/10) (n/r) \wgt(R)$. 
Since $R$ is sound, the latter is in $(1\pm\eps/4) \aN(d)$.
Also, note that 
\begin{equation}
\EX[Y_i | R] = \EX[{Y_1} | R] \geq \frac{q\wgt(R)}{2\hat{d}_R} \geq \frac{(r/n) \aN(d)}{4(100r(2m/n)}
= \frac{\aN(d)}{800m} \label{eqn:Yi}
\end{equation}
By linearity of expectation, $\EX[\sum_{i \leq q} Y_i | R] = q\EX[{Y_1}|R]$. Observe that $Y_i \leq 1/d$.
We can apply the Chernoff bound of \Thm{chernoff} to the iid random variables $(Y_i | R)$.
\begin{eqnarray}
& & \Pr[|\sum_i Y_i - \EX[\sum_i Y_i]| > (\eps/100)\EX[\sum_i Y_i] | R] \nonumber \\
& \leq & 2\exp\Big(-\frac{\eps^2}{3\cdot100^2} \cdot d \cdot q\EX[{Y_1} | R] \Big)
\label{eqn:edge}
\end{eqnarray}
We use \Eqn{Yi} to bound the (positive) term in the exponent is at least
$$ \frac{\eps^2}{3\cdot100^2} \cdot \frac{c\eps^{-2}m}{\aN(d)} \cdot \frac{\aN(d)}{800m} \geq 10. $$
Thus, if $R$ is sound, the following bound holds with probability at least $0.99$.
We also apply \Eqn{Yi-exp}.
\begin{eqnarray*}
\aN(d) & = & (n/r)(\hat{d}_R/q) \sum_{i=1}^q Y_i \\
&\in & (1\pm\eps/100) (n/r)(\hat{d}_R/q) q\EX[Y_i | R] \\
& \in & (1\pm\eps/100)(1\pm\eps/10) (n/r) \wgt(R) \in (1\pm \eps/4) \hNd
\end{eqnarray*}
The probability that $R$ is sound is at least $8/9$. A union bound completes the proof.
\end{proof}

The bounds on $r$ and $q$ in \Lem{edge} depend on the degree $d$. We now bring in the $h$ and $z$-indices
to derive bounds that hold for all $d$. We also remove the conditioning over a good $\Lambda$.

\begin{proof} (of \Thm{mainbound}) We will first assume that $\Lambda$ is good. 
By \Clm{lambda}, $\aN(d) \in [N((1+\eps/9)d, N((1-\eps/9)d)]$.

Suppose $\aN(d) = 0$, so there are no vertices with $\hat{d}_v \geq d$. By the bound above,
$N((1+\eps/9)d) = 0$, implying that $N((1+\eps/2)d) = 0$.
Furthermore $\hNd = 0$,
since the random variables $X_i$ and $Y_i$ in \mainalg{} can never be non-zero. Thus,
$\hNd = N((1+\eps/2)d)$, completing the proof.
%

We now assume that $\aN(d) > 0$. We split into two cases, depending on whether \Step{vertex} outputs 
or not. By \Clm{vertex}, with probability $> 9/10$, if \Step{vertex} outputs,
then $\hNd \in (1\pm \eps/9)\aN(d)$. 
By combining these bounds, the desired bound on $\hNd$ holds with probability $> 9/10$, conditioned on a good $\Lambda$.

Henceforth, we focus on the case that \Step{vertex} does not output. By \Clm{vertex},
$\aN(d) < 2c\eps^{-2}(n/r)$. By the choice of $r$ and \Clm{lambda}, $\aN((1+\eps/9)d) < \hindex$.
By the characterization of $\hindex$ of \Lem{hindex}, $\findex^2 \leq \max(\aN((1+\eps/9)d), (1+\eps/9)d) = (1+\eps/9)d$.
This implies that $r \geq c\eps^{-2}n/d$. By the definition of $\findex$, 
$\findex^2 \leq N(\min(d_{max},(1+\eps/9)d))\cdot \min(d_{max},(1+\eps/9)d)$.
By the \Clm{lambda} bound in the first paragraph, $\aN(d) \geq N((1+\eps/9)d)$. 
Since $\aN(d) > 0$, $\aN(d) \geq \aN(d_{max})$. Thus, $\findex^2 \leq \aN(d)\cdot(1+\eps/9)d$. 
%
and hence, $m \leq c\eps^{-2}m/(d\aN(d))$. The parameters satisfy the conditions in \Lem{edge}.
With probability $>7/8$, $\hNd \in (1\pm\eps/4)\aN(d)$, and by \Clm{lambda}, $\hNd$ has the desired
accuracy.

All in all, assuming $\Lambda$ is good, with probability at least $7/8$, $\hNd$ has the desired accuracy.
The conditioning on a good $\Lambda$ is removed by \Clm{cond} to complete the proof. 
\end{proof}

\section{Experimental Results} \label{sec:results}
We implemented our algorithm in C++ and performed our experiments on a MacBook Pro laptop with 2.7 GHz Intel Core i5 with 8 GB RAM. We performed our experiments on a collection of graphs from
SNAP~\cite{SNAP}, including social networks, web networks, and infrastructure networks.
The graphs typically have millions of edges, with the largest
having more than 100M edges. Basic properties of these graphs are presented in Table~\ref{tab:GraphStats}. 
We ignore direction  and treat all edges as undirected edges. 

\subsection{Implementation Details} \label{sec:implement}
For the HDM, we explicitly describe the procedure \degest$(v)$, which estimates  the degree of a given vertex $(v)$. 
\vspace{-10pt}
\begin{algorithm} 
\caption{\degest$(v)$} \label{alg:degest}
(Initialize $S = \emptyset$.) Repeatedly add u.a.r. vertex to $S$, until the number of pair-wise collisions
is at least $k = 25$. \\
Output ${|S|\choose 2}/k$ as estimate $\hat{d}_v$.
\end{algorithm}
In the algorithm \degest, a ``pair-wise collision" refers to a pair of neighbor samples that
yield the same vertex.
The expected number of pair-wise collisions is
${|S|\choose 2}/d_v$. We simply reverse engineer that inequality to get the estimate $\hat{d}_v$.
Ron and Tsur essentially prove that this estimate has low variance~\cite{RT16}.

\textbf{Setting the parameter values.} The boosting parameter $\rep$ is simply set to $1$.
(In some sense, we only introduced the median boosting for the theoretical union bound.
In practice, convergence is much more rapid that predicted by the Chernoff bound.)

The threshold $\thresh$ is set to $100$.
The parameters $r$ and $q$
are chosen to be typically around $0.005n$.
These are not ``sublinear" per se, but are an order of magnitude smaller than the queries
made in existing graph sampling results (more discussion in next section).

We set $D = \{\lfloor 1.1^i\rfloor\}$, since that gives a sufficiently fine-grained approximation
at all scales of the degree distribution.

%
%
Code for all experiments is available \href{https://sjain12@bitbucket.org/sjain12/saddles.git}{here}\footnote{https://sjain12@bitbucket.org/sjain12/saddles.git}.


\begin{table*}[]
\centering
\caption{\textbf{Graph properties}: \#vertices (n), \#edges (m), maximum degree, $h$-index and $z$-index. The last column indicates the median number of samples over 100 runs (as a percentage of $m$) required by \mainalg{} under HDM, with $r+q=0.01n$.}
\label{tab:GraphStats}

\begin{adjustbox}{max width=\textwidth}
\begin{tabular}{|l|l|l|r|r|r|r|c|}
\hline
 \multicolumn{1}{|c|}{}  &\multicolumn{1}{|c|}{}  &\multicolumn{1}{|c|}{}  & \multicolumn{1}{|c|}{\textbf{max.}}  &\multicolumn{1}{|c|}{\textbf{avg.}}  &\multicolumn{1}{|c|}{}  &\multicolumn{1}{|c|}{}  &\multicolumn{1}{|c|}{\textbf{Perc. edge } }  \\
 \textbf{graph}         & \textbf{\#vertices}     & \textbf{\#edges}        & \textbf{degree}  & \textbf{degree}       & \textbf{H-index}   & \textbf{Z-index} & \textbf{ samples for HDM}  \\ \cline{7-7}
\hline
loc-gowalla            & 1.97E+05              & 9.50E+05               & 14730    & 4.8      &   275      & 101  & 7.0        \\
web-Stanford           & 2.82E+05              & 1.99E+06          & 38625      & 7.0        &  427    & 148    & 6.4   \\
com-youtube            & 1.13E+06              & 2.99E+06                  & 28754    & 2.6    &   547       &121    & 11.7         \\
web-Google             & 8.76E+05              & 4.32E+06                 & 6332      & 4.9     &  419    & 73   & 6.2     \\
web-BerkStan           & 6.85E+05              & 6.65E+06                & 84230     & 9.7     &  707   & 220  & 5.5   \\
wiki-Talk               &  2.39E+06                 &  9.32E+06             &  100029      & 3.9           &  1055      & 180       & 8.5  \\
as-skitter             & 1.70E+06              & 1.11E+07             & 35455      & 6.5       & 982       &184     & 6.7     \\
cit-Patents            & 3.77E+06              & 1.65E+07               & 793      & 4.3      &  237     & 28      & 5.6   \\
com-lj                 & 4.00E+06              & 3.47E+07                  & 14815     & 8.6      &  810       & 114  & 4.7  \\
soc-LiveJournal1    &  4.85E+06                 &  8.57E+07         & 20333        & 17.7        &   989      &124      & 2.4 \\
com-orkut              & 3.07E+06              & 1.17E+08                  & 33313   & 38.1      &   1638    & 172    & 2.0     \\  
\hline
\end{tabular}
\end{adjustbox}
\end{table*}

\subsection{Evaluation  of \mainalg{}} \label{sec:eval}


{\bf Accuracy over all graphs:} We run \mainalg{} with the parameters discussed above for a variety of graphs. 
Because of space considerations, we do not show results for all graphs in this version.
(We discovered the results to be consistent among all our experiments.)
\Fig{intro} show the results for the SM for some graphs in \Tab{GraphStats}.
For all these runs, we set $r+q$ to be 1\% of the number of vertices in the graph.
Note that the sample size of \mainalg{} in the SM is exactly $r+q$. 
For the HDM, we show results in \Fig{runs-hdm}. 
Again, we set $r+q$ to be 1\%, though the number of edges sampled (due to invocations of \degest$(v)$) varies quite a bit.
The required number of samples are provided in \Tab{GraphStats}. Note that the number of edges sampled
is well within 10\% of the total, except for the {\tt com-youtube} graph.

%

Visually, we can see that the estimates are accurate for all degrees, in all graphs, for both models.
This is despite there being sufficient irregular behavior in $N(d)$. Note that the shape of the various ccdhs are different and none of them form an obvious straight line.
Nonetheless, \mainalg{} captures the distribution almost perfectly in all cases by observing 1\% of the vertices.


%
\begin{figure*}[ht!]
\centering
    \subfigure[{\tt as-skitter}]{ 
    \includegraphics[width=0.2\textwidth]{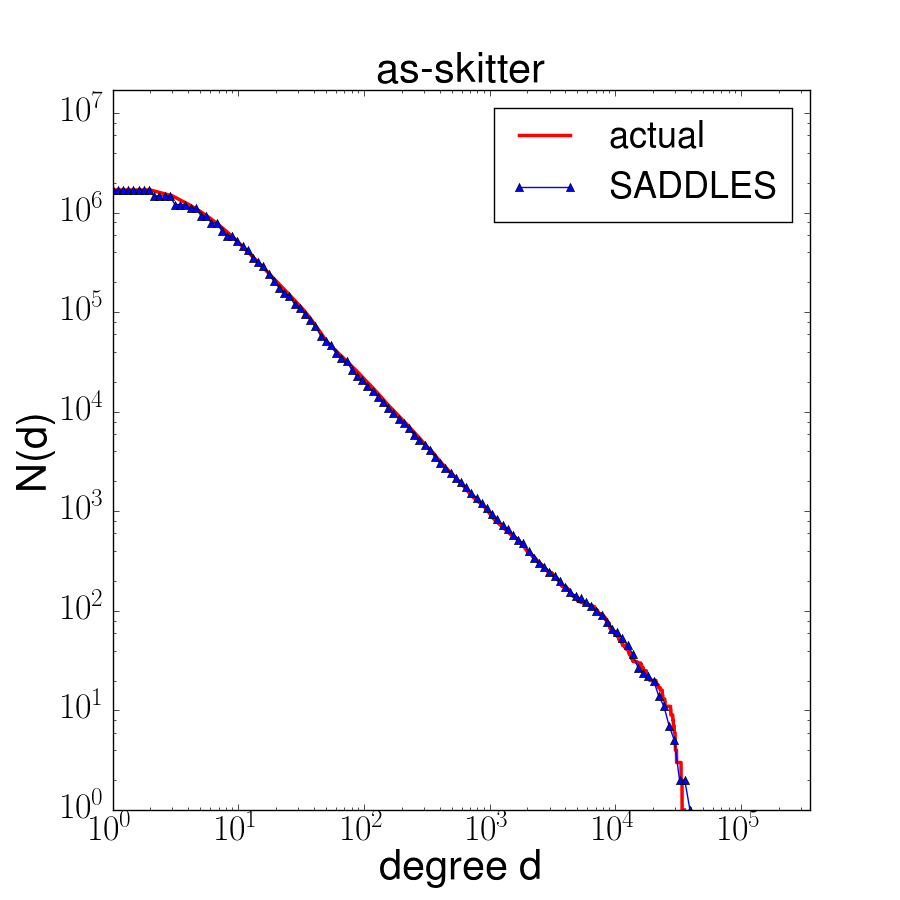} \label{fig:skitter-hdm}}
    \subfigure[{\tt loc-gowalla}]{
    \includegraphics[width=0.2\textwidth]{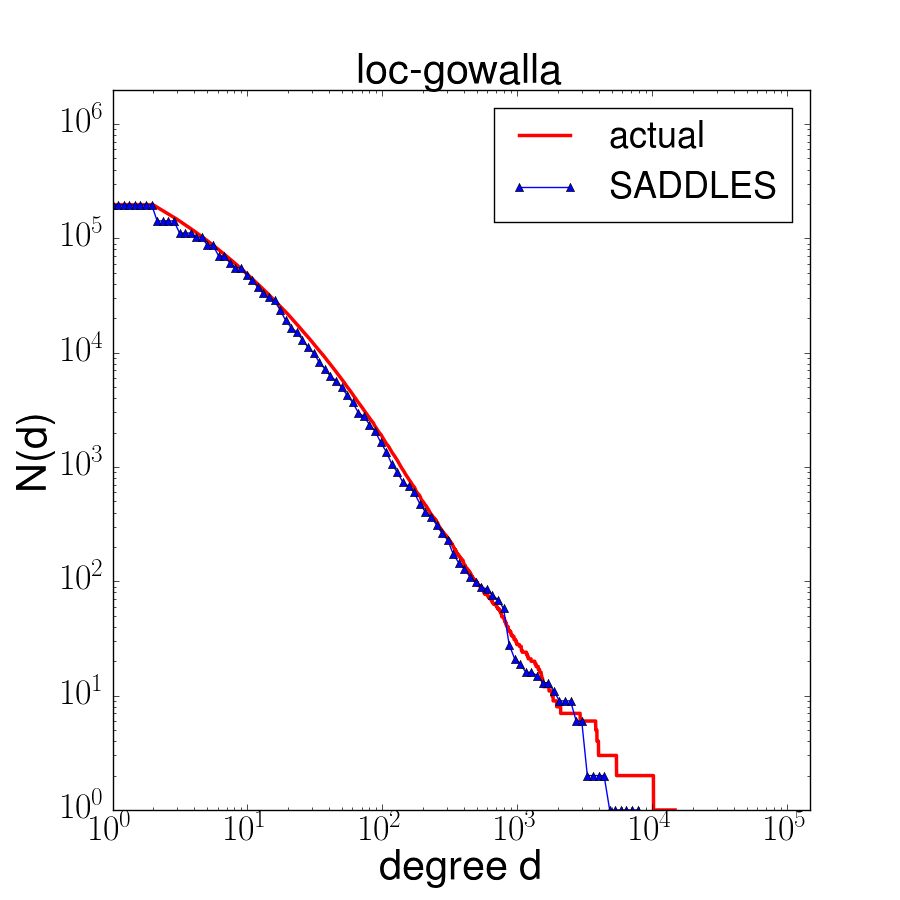} \label{fig:gowalla-hdm}}
    \subfigure[{\tt web-Google}]{
    \includegraphics[width=0.2\textwidth]{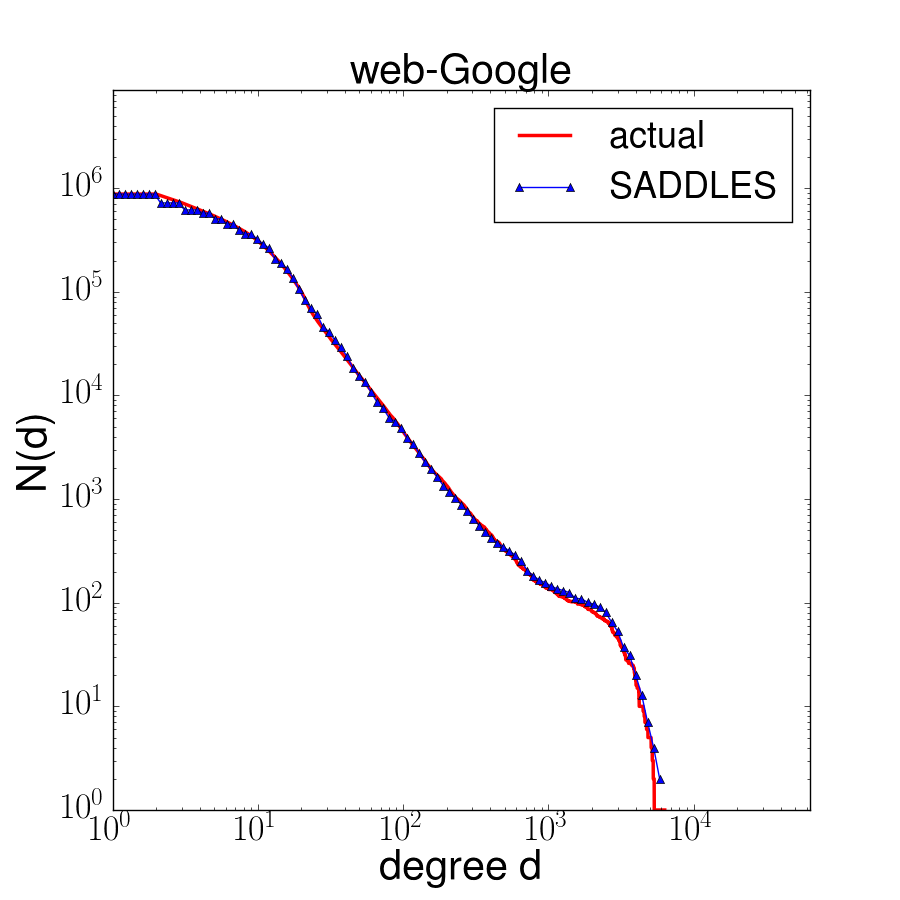} \label{fig:google-hdm}}
    \subfigure[{\tt wiki-Talk}]{
    \includegraphics[width=0.2\textwidth]{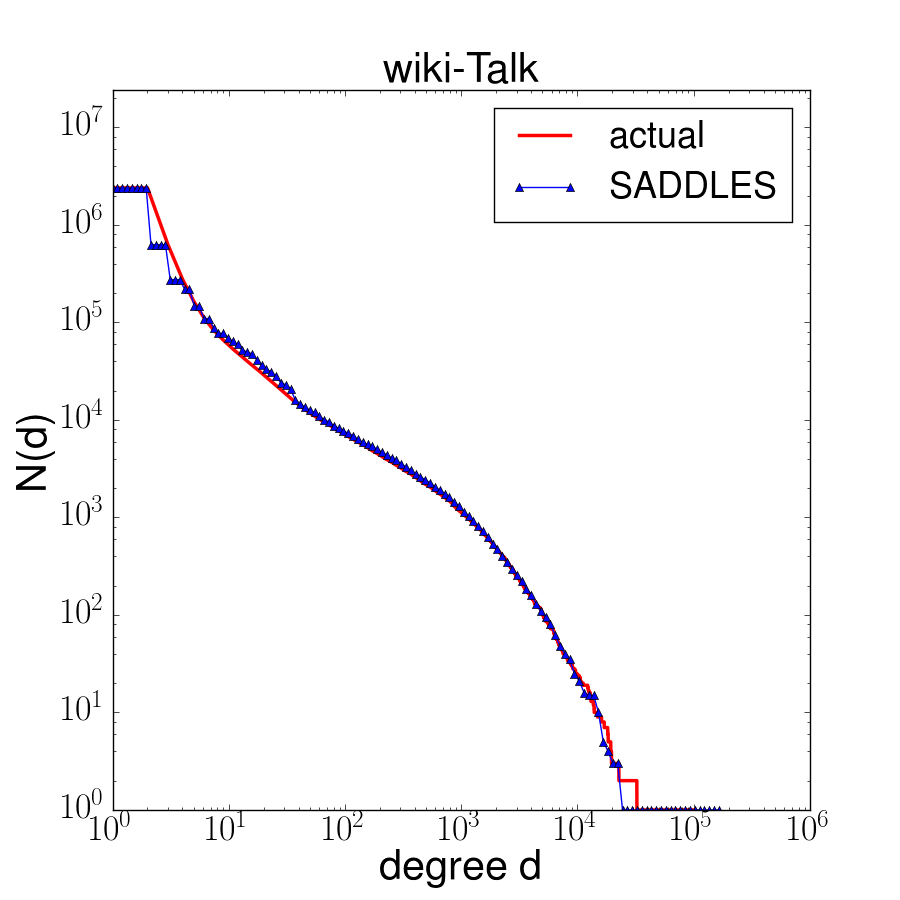} \label{fig:wiki-hdm}}
    \vspace{-10pt}
\caption{The result of runs of \mainalg{} on a variety of graphs, for the HDM. We set $r+q$ to be 1\% of the number
of vertices, for all graphs. The actual number of edges sampled varies, and is given in \Tab{GraphStats}.}
\label{fig:runs-hdm}        
\end{figure*}

{\bf Convergence:} To demonstrate convergence, we fix the graph {\tt com-orkut}, and run \mainalg{}
only for the degrees $10, 100$, and  $1000$. For each choice of degree, we vary the total number of samples $r+q$. (We set $r=q$ in all runs.) Finally, for each
setting of $r+q$, we perform 100 independent runs of \mainalg.

For each such run, we compute an error parameter $\alpha$. Suppose the output of a run is $M$,
for degree $d$. The value of $\alpha$ is the smallest value of $\epsilon$,
such that $M \in [(1-\epsilon)N((1+\epsilon)d), (1+\epsilon)N((1-\epsilon)d)]$. (It is the smallest
$\epsilon$ such that $M$ is an $(\epsilon,\epsilon)$-approximation of $N(d)$.)

\Fig{conv} shows the spread of $\alpha$, for the $100$ runs, for each choice of $r+q$.
Observe how the spread decreases as $r+q$ goes to 10\%. In all cases, the values of $\alpha$
decay to less than $0.05$. We notice that convergence is much faster for $d=10$. This is because
$N(10)$ is quite large, and \mainalg{} is using vertex sampling to estimate the value. 

\begin{figure*}[ht]
\centering
    \subfigure[{$d=10$}]{\includegraphics[width=0.2\textwidth]{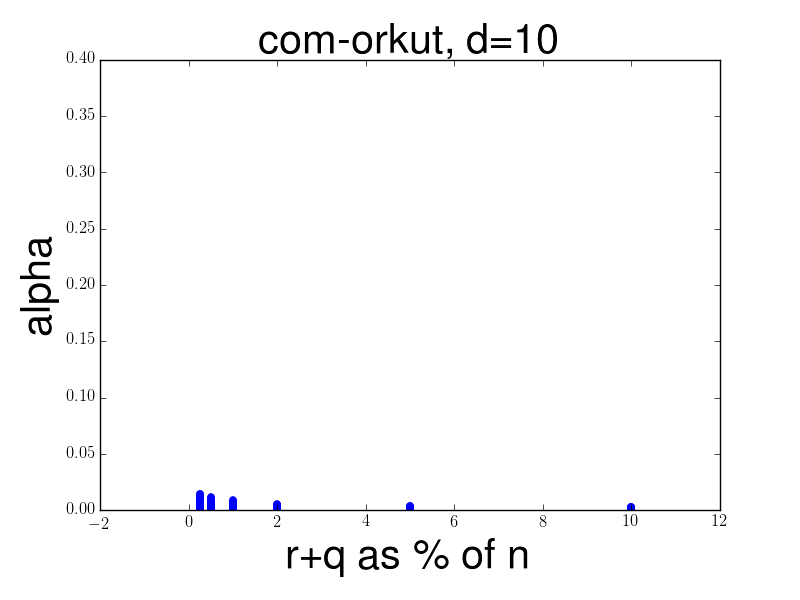} \label{fig:conv-10}}
    \subfigure[{$d=100$}]{\includegraphics[width=0.2\textwidth]{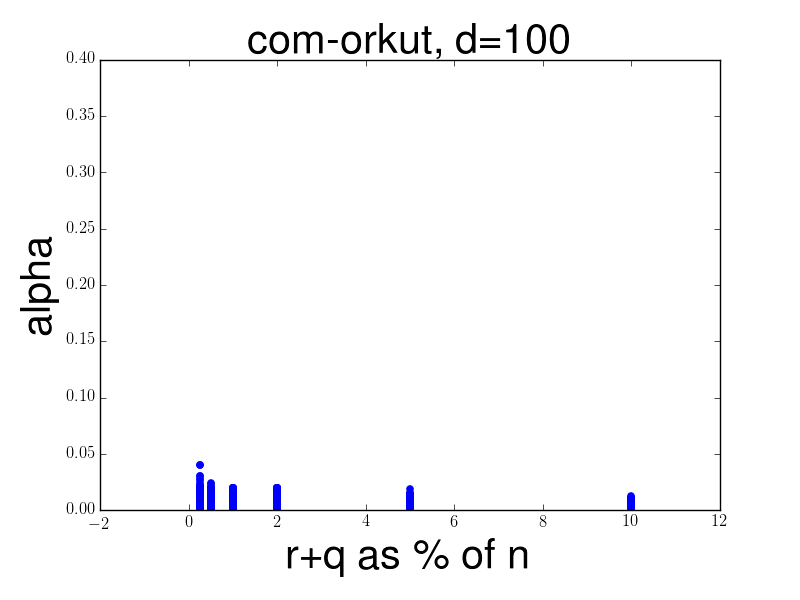} \label{fig:conv-100}}
    \subfigure[{$d=1000$}]{
    \includegraphics[width=0.2\textwidth]{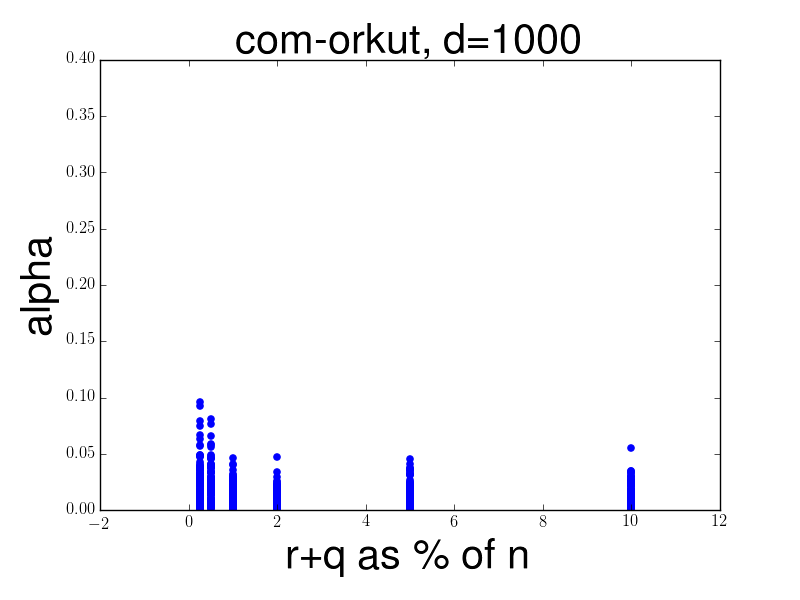} \label{fig:conv-1000}}
    \subfigure[{$d=10000$}]{
    \includegraphics[width=0.2\textwidth]{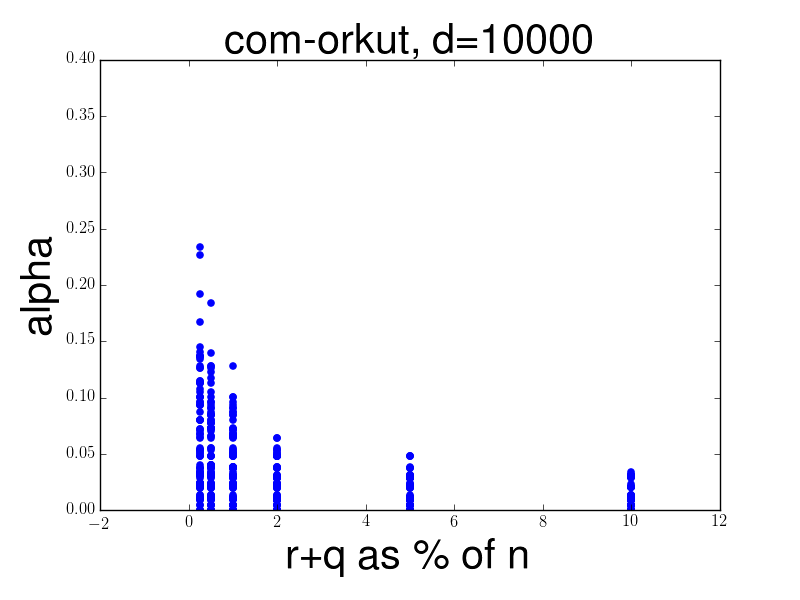} \label{fig:conv-10000}}
    \vspace{-10pt}
\caption{Convergence of \mainalg: We plot the values of the error parameter $\alpha$ (as defined in \Sec{eval}) for 100 runs at increasing values of $r+q$.
We have a different plot for $d = 10, 100, 1000, 10000$ to show the convergence at varying portions of the ccdh.}
\label{fig:conv}        
\end{figure*}

{\bf Large value of $h$ and $z$-index on real graphs:} The $h$ and $z$-index of all graphs is given in \Tab{GraphStats}. Observe how they
are typically in the hundreds. Note that the average degree is typically an order of magnitude smaller than these indices.
Thus, a sample size of $n/h + m/z^2$ (as given by \Thm{main}, ignoring constants) is significantly sublinear.
This is consistent with our choice of $r+q = n/100$ leading to accurate estimates for the ccdh.

\subsection{Comparison with previous work}

There are several graph sampling algorithms that have been discussed in ~\cite{RT12,LKJ06,EHR08,LF06,ANK10,PSE15,ZKS15}. In all of these methods we collect the vertices and scale their counts appropriately to get the estimated ccdh.
We describe these methods below in more detail, and discuss our implementation of the method.

\begin{asparaitem}
\item Vertex Sampling (VS, also called egocentric sampling) ~\cite{RT12,LKJ06,EHR08,LF06,PSE15,ANK14}: In this algorithm, we  sample vertices u.a.r. and scale the ccdh obtained appropriately, to get an estimate for the ccdh of the entire graph. 
\item Edge Sampling (ES) ~\cite{RT12,LKJ06,EHR08,LF06,PSE15,ANK14}: This algorithm samples edges u.a.r. and includes one or both end points in the sampled network. Note that
this does \emph{not} fall into the SM. In our implementation we pick a random end point. 
\item Random walk with jump (RWJ) ~\cite{RT12,EHR08,LF06,PSE15,ANK14}: We start a random walk at a vertex selected u.a.r. and collect all vertices encountered on the path in our sampled network. 
At any point, with a constant probability ($0.15$, based on previous results) we jump to another u.a.r. vertex. 
\item One Wave Snowball (OWS) ~\cite{LKJ06,EHR08,ANK14}:  Snowball sampling starts with some vertices selected u.a.r. and crawls the network until a network of the desired size is sampled. 
In our implementation, we usually stop at the first level since that accumulates enough vertices.
\item Forest fire (FF)~\cite{EHR08,LF06,ANK14}: This method generates random sub-crawls of the network.
A vertex is picked u.a.r. and randomly selects a subset of its neighbors
(according to a geometric distribution).
The process is repeated from every selected vertex until it ends. It is then repeated from another u.a.r. vertex.
\end{asparaitem}

We run all these algorithms on the {\tt amazon0601}, {\tt web-Google},
{\tt cit-Patents}, and {\tt com-orkut} networks.
To make fair comparisons, we run each method until it selects 1\% of the vertices.
The comparisons are shown in \Fig{intro}. Observe how none of the methods come close to accurately measuring
the ccdh. (This is consistent with previous work, where typically 10-20\% of the vertices are sampled
for results.) Naive vertex sampling is accurate at the head of the distribution, but completely misses the tail.
Except for vertex sampling, all other algorithms are biased towards the tail. Crawls find high degree vertices
with disproportionately higher probability, and overestimate the tail.

Note that our implementations of FF, OWS, RWJ assume access to u.a.r. vertices. Variants of these algorithms can be used in situations where we only have access to seed vertices, however, one would typically have to sample many more edges to deal with larger correlation among the vertices obtained through the random walks. Despite this extra capability to sample u.a.r. vertices in our implementation of these algorithms, they show significant errors, particularly in the tail of the distribution.

{\bf Inverse method of Zhang et al~\cite{ZKS15}:} An important result of estimating degree distributions is that of Zhang et al~\cite{ZKS15}, that explicitly points
out the bias problems in various sampling methods. They propose
a bias correction method by solving a constrained, penalized weighted least-squares problem
on the sampled degree distribution.
We apply this method for the sampling methods demonstrated in their paper, namely VS, OWS,  and IN (sample vertices u.a.r. and only retain edges between sampled vertices). 
We show results in \Fig{intro}, again with a sample size of 1\% of the vertices. Observe that 
no method gets even close to estimating the ccdh accurately, even after debiasing. Fundamentally, these methods require
significantly more samples to generate accurate estimates.

The running time and memory requirements of this method grow superlinearly with the
maximum degree in the graph.  
The largest graph processed by~\cite{ZKS15} has a few hundred thousand edges, which is on the smaller side of graphs in \Tab{GraphStats}. \mainalg{} processes a graph with more than 100M edges in less than a minute, while our attempts to run the~\cite{ZKS15} algorithm on this graph did not terminate in hours.

\vspace{-10pt}
\section{Acknowledgements} \label{sec:ack}

Ali Pinar's work is supported by the Laboratory Directed Research and Development program at Sandia National Laboratories.
Sandia National Laboratories is a multimission laboratory managed and
operated by National Technology and Engineering Solutions of Sandia,
LLC., a wholly owned subsidiary of Honeywell International, Inc., for
the U.S. Department of Energy's National Nuclear Security Administration
under contract DE-NA-0003525.

Both Shweta Jain and C. Seshadhri are grateful to the support of the Sandia National Laboratories LDRD program for funding this research.
C. Seshadhri also acknowledges the support of NSF TRIPODS grant, CCF-1740850.

This research was partially supported by the Israel Science Foundation grant No.~671/13 and by a grant from the Blavatnik fund. 
Talya Eden is grateful to the Azrieli Foundation for the award of an Azrieli Fellowship.

Both Talya Eden and C. Seshadhri are grateful to the support of the Simons Institute, where this work was initiated during
the Algorithms and Uncertainty Semester.

\bibliographystyle{ACM-Reference-Format}
\bibliography{dd_kdd}

\newpage

%
\end{document}